\documentclass[runningheads]{llncs}
\usepackage{hyperref,breakurl}
\usepackage{bm}
\usepackage{stmaryrd}
\usepackage{amssymb}
\usepackage{amsmath}
\usepackage{enumerate}
\usepackage{color}
\usepackage{tikz}
\usetikzlibrary{arrows,automata,chains,matrix,positioning,scopes,calc,decorations.pathmorphing,shapes.misc}

\tikzset{drbox/.style={
  rounded rectangle,
  minimum size=2mm,
  very thick,draw=black!50,
  top color=white,bottom color=black!20,
  font=\ttfamily}}
\tikzset{prbox/.style={
  rounded rectangle,
  minimum size=2mm,
  very thick,draw=white,
  top color=white,bottom color=white,
  font=\ttfamily}}
\tikzset{mode/.style={font=\scriptsize,execute at begin node=$,execute
    at end node=$}}

\DeclareSymbolFont{frenchscript}{OMS}{ztmcm}{m}{n}
\DeclareMathSymbol{\fL}{\mathord}{frenchscript}{76}  
\DeclareMathSymbol{\fT}{\mathord}{frenchscript}{84}  
\DeclareMathSymbol{\V}{\mathord}{frenchscript}{88}   
\spnewtheorem{observation}{Observation}{\bfseries}{\itshape}

\newcommand{\df}[1]{Definition~\ref{df:#1}}
\newcommand{\thm}[1]{Theorem~\ref{thm:#1}}

\newcommand{\lem}[1]{Lemma~\ref{lem:#1}}
\newcommand{\obs}[1]{Observation~\ref{obs:#1}}

\newcommand{\sect}[1]{Section~\ref{sec:#1}}
\newcommand{\tab}[1]{Table~\ref{tab:#1}}

\newfont{\bbb}{bbm10 scaled 1100}        
\newcommand{\T}{\mathcal{P}}             
\newcommand{\bbT}{\T}
\newcommand{\pT}{T}                      
\newcommand{\fTHT}{\fT_{\textrm{HT}}}       
\newcommand{\N} {{\cal N}}               
\newcommand{\n}{{\it n}}                 
\newcommand{\bn}{{\it bn}}               
\newcommand{\fn}{{\it fn}}               

\newcommand{\subs}[2]{\{\mathord{\raisebox{2pt}[0pt]{$#1$}\!/\!#2}\}} 
\renewcommand{\nu}{}
\newcommand{\plat}[1]{\raisebox{0pt}[0pt][0pt]{#1}}     
\newcommand\ie{i.\,e.~}
\newcommand\eg{e.\,g.~}

\newcommand\cf{cf.~}

\newcommand{\trans}[1]{\fT(#1)}
\newcommand{\subtrans}{\fT}
\newcommand{\op}{\mathtt{op}}
\newcommand{\nil}{\mathbf{0}}
\newcommand{\equred}{\mathrel{\equiv}}
\newcommand{\transtau}{\mathrel{\raisebox{0pt}[8pt][0pt]{$\stackrel{\tau}{\longrightarrow}$}^{\!*\,}}}

\newcommand{\wesim}{\ensuremath{\mathrel{\approx_{\mathrm{EWB}}}}}
\newcommand{\wbb}{\stackrel{\raisebox{-1pt}[0pt][0pt]{$\scriptscriptstyle \bullet$}}{\approx}}
\newcommand{\wssbb}{\wbb^\surd}
\newcommand{\wssbbd}{\wbb^{\!\surd\!\Delta}\!}
\newcommand{\wbbisim}{\ensuremath{\mathrel{\approx_{\mathrm{AWBB}}}}}
\newcommand{\wotau}{\ensuremath{\mathrel{\approx_{\mathrm{W}o\tau}}}}
\newcommand{\wab}{\ensuremath{\mathrel{\approx_{\mathrm{WAB}}}}}
\newcommand{\wlbsim}{\ensuremath{\mathrel{\approx_{\mathrm{WCB}}}}}

\newcommand{\wcomm}{\ocomm}
\newcommand{\scomm}{\sbarb}
\newcommand{\sbarb}[1]{{\downarrow_{#1}}}
\newcommand{\ocomm}[1]{{\Downarrow_{#1}}}
\newcommand{\ascomm}[1]{\downarrow^{\mathrm{c}}_{#1}}
\newcommand{\acomm}[1]{\Downarrow^{\mathrm{c}}_{#1}}

\newcommand{\mrel}{\ensuremath{\mathrel{\mathcal{R}}}}

\newcommand\piCal{$\pi$-calculus}
\newcommand{\pims}{\pi}
\newcommand{\pima}{{\rm a\pi}}

\newcommand{\Act}{\mathrm{Act}}

\newcommand{\sI}{\fT}

\begin{document}
\title{Stronger Validity Criteria\texorpdfstring{\\}{} for Encoding Synchrony\thanks{This work was partially supported by the DFG (German Research Foundation).}}
\author{Rob van Glabbeek\inst{1,2}, Ursula Goltz\inst{3},
        Christopher Lippert\inst{3}, Stephan Mennicke\inst{3}}
\authorrunning{R.J. van Glabbeek, U. Goltz, Ch. Lippert and S. Mennicke}
\institute{Data61, CSIRO, Sydney, Australia \and
  Computer Sc.\ and Engineering,
  University of New South Wales, Sydney, Australia \and
  Institute for Programming and Reactive Systems, TU Braunschweig, Germany}

\maketitle
\setcounter{footnote}{0}

\begin{abstract}
We analyse two translations from the synchronous into the asynchronous $\pi$-calculus, both without choice,
that are often quoted as standard examples of valid encodings, showing that the asynchronous
$\pi$-calculus is just as expressive as the synchronous one.
We examine which of the quality criteria for encodings from the literature support the validity
of these translations. Moreover, we prove their validity according to much stronger
criteria than considered previously in the literature.

\keywords{Process calculi \and expressiveness \and translations \and quality criteria for encodings \and
 valid encodings \and compositionality \and operational correspondence \and semantic equivalences \and asynchronous $\pi$-calculus.}
\vspace{3ex}

\emph{This paper is dedicated to Catuscia Palamidessi, on the occasion of her birthday.
It has always been a big pleasure and inspiration to discuss with her.}
\end{abstract}

\section{Introduction}

In the literature, many definitions are proposed of what it means for one system description
language to encode another one. Each concept $C$ of a valid encoding yields an ordering of system
description languages with respect to expressive power: language $\fL'$ is \emph{at least as expressive as}
language $\fL$ (according to $C$), notation $\fL \preceq_C \fL'$, iff a valid encoding from $\fL$ to $\fL'$ exists.
The concepts of a valid encoding themselves, the \emph{validity criteria}, also can be ordered:
criterion $C$ is \emph{stronger} than criterion $D$ iff for each two system
description languages $\fL$ and $\fL'$ one has\vspace{-1ex}
            $$\fL \preceq_C \fL' ~~\Rightarrow~~ \fL \preceq_D \fL'\;.$$
Naturally, employing a stronger validity criterion constitutes a stronger claim that the
target language is at least as expressive as the source language.

In this paper, we analyse two well-known translations from the synchronous into the asynchronous
$\pi$-calculus, one by Boudol and one by Honda \& Tokoro. Both are often quoted as standard examples of
valid encodings. We examine which of the validity criteria from the literature support the validity
of these encodings. Moreover, we prove the validity of these encodings according to much stronger
criteria than considered previously in the literature.
\pagebreak[3]

A \emph{translation} $\fT$ from (or \emph{encoding} of) a language $\fL$ into a language $\fL'$ is a
function from the $\fL$-expressions to the $\fL'$-expressions. The first formal definition of
a \emph{valid} encoding of one system description language into another stems from Boudol
\cite{Bo85}. It is parametrised by the choice of a semantic equivalence $\sim$ that is meaningful for the
source as well as the target language of the translation---and is required to be a congruence for
both. Boudol in particular considers languages whose semantics are given in terms of labelled
transition systems. Any semantic equivalence defined on labelled transition systems, such as
strong bisimilarity, induces an equivalence on the expressions of such languages, and thus allows
comparison of expressions from different languages of this kind. Boudol formulates two requirements
for valid translations: (1) they should be \emph{compositional}, and (2) for each source language expression $P$,
its translation $\fT(P)$---an expression in the target language---is semantically equivalent to $P$.
\advance\textheight 1pt

\advance\textheight -1pt
Successive generalisations of the definition of a valid encoding from Boudol \cite{Bo85} appear in \cite{vG94a,vG12,vG18b}.
These generalisations chiefly deal with languages that feature process variables, and that are
interpreted in a semantic domain (such as labelled transition systems) where not every semantic value
need be denotable by a closed term. The present paper, following \cite{Bo85} and most of the
expressiveness literature, deals solely with \emph{closed-term} languages,
in which the distinction between syntax and semantic is effectively dropped by taking the domain of
semantic values, in which the language is interpreted, to consist of the closed terms of the language.
In this setting the only generalisation of the notion of a valid encoding from
\cite{vG94a,vG12,vG18b} over \cite{Bo85} is that Boudol's congruence requirement on the semantic equivalence 
up to which languages are compared is dropped. In \cite{vG18b} it is also shown that the requirement of
compositionality can be dropped, as in the presence of process variables it is effectively implied
by the requirement that semantic equivalence is preserved upon translation. But when dealing with
languages without process variables, as in the present paper, it remains necessary to require
compositionality separately.

A variant of the validity criterion from Boudol is the notion of \emph{full abstraction}, employed
in \cite{Rie91,Sha91,Sha92,NestmannP00,Nestmann00,BPV05,Fu16}. In this setting, instead of a single
semantic equivalence $\sim$ that is meaningful for the source as well as the target language of the
translation, two semantic equivalences $\sim_{\rm S}$ and $\sim_{\rm T}$ are used as parameters of
the criterion, one on the source and one on the target language. Full abstraction requires, for
source expressions $P$ and $Q$, that
$P \sim_{\rm S} Q \Leftrightarrow \fT(P) \sim_{\rm T} \fT(Q)$.
Full abstraction has been criticised as a validity criteria for encodings in \cite{BPV08,GN16,Parrow16};
a historical treatment of the concept can be found in \cite[Sect.~18]{vG18e}.

An alternative for the equivalence-based validity criteria reviewed above are the ones employing
\emph{operational correspondence}, introduced by Nestmann \& Pierce in \cite{NestmannP00}.
Here valid encodings are required to satisfy various criteria, differing subtly from paper to
paper; often these criteria are chosen to conveniently establish that a given language is or is not
as least as expressive as another. Normally some form of operational correspondence is one of these
criteria, and as a consequence of this these approaches are suitable for comparing the
expressiveness of process calculi with a reduction semantics, rather than system description
languages in general. Gorla \cite{Gorla10a} has selected five of these criteria as a
unified approach to encodability and separation results for process calculi---\emph{compositionality},
\emph{name invariance}, \emph{operational correspondence}, \emph{divergence reflection} and
\emph{success sensitiveness}---and since then these criteria have been widely accepted as
constituting a standard definition of a valid encoding.
\advance\textheight 1pt

\advance\textheight -1pt
In \cite{Palamidessi03} Catuscia Palamidessi employs four requirements for valid encodings between
languages that both contain a parallel composition operator $|$:
compositionality, preservation of semantics, a form of name invariance, and the requirement that
parallel composition is translated homomorphically, i.e., $\fT(P|Q) = \fT(P) | \fT(Q)$.
The latter is not implied by any of the requirements considered above.
The justification for this requirement is that it ensures that the translation maintains the degree
of distribution of the system. However, Peters, Nestmann \& Goltz \cite{PNG13} argue that it is possible to
maintain the degree of distribution of a system upon translation without requiring a homomorphic
translation of $|$; in fact they introduce the criterion \emph{preservation of distributability}
that is weaker then the homomorphic translation of $|$.

This paper analyses the encodings $\fT_{\rm B}$ and $\fTHT$ of Boudol and Honda \&
Tokoro of the synchronous into the asynchronous $\pi$-calculus, both without the choice operator $+$.
Our aim is to evaluate the validity of these encodings with respect to all criteria for valid
encodings summarised above.

\sect{pi2api} recalls the encodings $\fT_{\rm B}$ and $\fT_{\rm HT}$.
\sect{Gorla-validity} reviews the validity criteria from Gorla \cite{Gorla10a}, and recalls the result
from \cite{vG18a} that the encodings $\fT_{\rm B}$ and $\fT_{\rm HT}$ meet all those criteria.
Trivially, $\fT_{\rm B}$ and $\fT_{\rm HT}$ also meet Palamidessi's criterion that
parallel composition is translated homomorphically, and thus also the criterion on preservation of
distributability from \cite{PNG13}.

\sect{compositionality} focuses on the criterion of compositionality. Gorla's proposal involves a weaker form of
this requirement, exactly because encodings like $\fT_{\rm B}$ and $\fTHT$ do not satisfy
the default form of compositionality. However, we show that these encodings also satisfy a form of
compositionality due to \cite{vG12} that significantly strengthens the one from \cite{Gorla10a}.
Moreover, depending on how the definition of valid encodings between concrete languages generalises
to one between parametrised languages, one may even conclude that $\fT_{\rm B}$ and $\fT_{\rm HT}$
satisfy the default notion of compositionality.

\sect{operational correspondence} focuses on the criterion of operational correspondence.
In \cite{NestmannP00} two forms of this criterion were proposed, one for \emph{prompt} and one for
\emph{nonprompt} encodings. Gorla's form of operational correspondence \cite{Gorla10a} is the natural common weakening
of the forms from Nestmann \& Pierce \cite{NestmannP00}, and thus applies to prompt as well as
nonprompt encodings. As the encodings $\fT_{\rm B}$ and $\fT_{\rm HT}$ are nonprompt, they certainly
do not meet the prompt form of operational correspondence from \cite{NestmannP00}. In \cite{vG18a}
it was shown that they not only satisfy the form of \cite{Gorla10a}, but even the nonprompt form
from \cite{NestmannP00}.

Gorla's form of operational correspondence, as well as the nonprompt form of \cite{NestmannP00},
weakens the prompt form in two ways. In \cite{vG18a} a natural intermediate form was contemplated
that weakens the prompt form in only one of these ways, and the open question was raised whether 
$\fT_{\rm B}$ and $\fT_{\rm HT}$ satisfy this intermediate form of operational correspondence.
The present paper answers that question affirmatively.

Gorla's criterion of success sensitiveness is a more abstract form of \emph{barb sensitiveness}.
The original barbs were predicates telling whether a process could input or output data over a
certain channel. In \sect{barbs} we show that whereas $\fT_{\rm B}$ is barb sensitive, $\fT_{\rm HT}$
is not. The encoding $\fT_{\rm HT}$ becomes barb sensitive if we use a weaker form of barb,
abstracting from the difference between input and output. This, however, is against the spirit of the
asynchronous $\pi$-calculus, where instead one abstracts from input barbs altogether.
Gorla's  criterion of success sensitiveness thus appears to be an improvement over barb sensitiveness.

\sect{validity upto} evaluates $\fT_{\rm B}$ and $\fT_{\rm HT}$ under the original validity criterion
of Boudol \cite{Bo85}, as generalised in \cite{vG12}; we call a compositional encoding $\fT$ \emph{valid up to}
a semantic equivalence ${\sim}$ iff $\trans{P} \sim P$ for all source language expressions $P$.
We observe that the encodings $\fT_{\rm B}$ and $\fT_{\rm HT}$ are not valid under equivalences that
match transition labels, such as early weak bisimilarity, nor under asynchronous weak bisimilarity.
Then we show that $\fT_{\rm B}$, but not $\fTHT$, is valid under weak barbed bisimilarity. This is
our main result. Finally, we introduce a new equivalence under which $\fTHT$ is valid: a version of
weak barbed bisimilarity that drops the distinction between input and output barbs.

\sect{versus} starts with the result that $\fT_{\rm B}$ and $\fT_{\rm HT}$ are both valid under a
version of weak barbed bisimilarity where an abstract success predicate takes over the role of barbs.
That statement turns out to be equivalent to the statement that these encodings are success sensitive
and satisfy a form of operational correspondence that is stronger then Gorla's. One can also
incorporate Gorla's requirement of divergence reflection into the definition of form of barbed bisimilarity.
Finally, we remark that $\fT_{\rm B}$ and $\fT_{\rm HT}$ remain valid when upgrading weak to branching bisimilarity.

\sect{full abstraction} applies a theorem from \cite{vG18e} to infer from the validity of
$\fT_{\rm B}$ and $\fT_{\rm HT}$ up to a form of weak barbed bisimilarity, that these encodings are
also fully abstract, when taking as source language equivalence weak barbed congruence, and as
target language equivalence the congruence closure of that form of weak barbed bisimilarity for the
image of the source language within the target language.

\section{Encoding Synchrony into Asynchrony}\label{sec:pi2api}

\noindent
Consider the $\pi$-calculus as presented by Milner in \cite{Mi92}, i.e., the one of
Sangiorgi and Walker \cite{SW01book} without matching, $\tau$-prefixing and choice.

Given a set of \emph{names} $\N$, the set $\mathcal{P}_\pi$ of \emph{processes} or \emph{terms} $P$
of the calculus is given by\vspace{-2ex}  $$P ::= \textbf{0}  ~~\mid~~ \bar xz.P ~~\mid~~ x(y).P ~~\mid~~ P|Q ~~\mid~~ (\nu y)P ~~\mid~~ !P$$
with $x,y,z,u,v,w$ ranging over $\N$.

$\bm{0}$ denotes the empty process.
$\bar{x}z$ stands for an output guard that sends the name $z$ along the channel $x$.
$x(y)$ denotes an input guard that waits for a name to be transmitted along the channel named $x$.
Upon receipt, the name is substituted for $y$ in the subsequent process.
$P|Q$ ($P, Q \in \pims$) denotes a parallel composition between $P$ and $Q$.
$!P$ is the replication construct and $(y)P$ restricts the scope of name $y$ to $P$.

\begin{definition}\rm\label{df:structural congruence}
An occurrence of a name $y$ in $\pi$-calculus process $P\in\T_\pi$ is \emph{bound} if it lies
within a subexpression $x(y).Q$ or $(\nu y)Q$ of $P$; otherwise it is \emph{free}.
Let $\n(P)$ be the set of names occurring in $P\in \T_\pi$,
and $\fn(P)$ (resp.\ $\bn(P)$) be the set of names occurring free (resp.\ bound) in $P$.

\emph{Structural congruence}, $\equred$, is the smallest congruence relation on processes satisfying\vspace{-2ex}
\[
\begin{array}[b]{@{}l@{\;\;}r@{~\!\equred\!~}l@{\hspace{10pt}}r@{~\!\equred\!~}l@{\;\;}r@{}}
\scriptstyle (1)& P | (Q | R) & (P | Q) | R &
(\nu y) \textbf{0} & \textbf{0} & \scriptstyle (5)\\

\scriptstyle (2)& P | Q & Q | P &
(\nu y)(\nu u)P & (\nu u)(\nu y) P & \scriptstyle (6)\\

\scriptstyle (3)& P | \textbf{0} & P &
(\nu w) (P | Q) & P | (\nu w)Q  & \scriptstyle (7)\\

& \multicolumn{2}{c}{} &
(\nu y)P & (\nu w)P\subs{w}{y} & \scriptstyle (8)\\

\scriptstyle (4)& !P & P | !P &
x(y).P & x(w).P\subs{w}{y}\;.  & \scriptstyle (9)\\
\end{array}
\]
Here $w\notin\n(P)$, and $P\subs{w}{y}$ denotes the process
obtained by replacing each free occurrence of $y$ in $P$ by $w$.
Rules (8) and (9) constitute \emph{$\alpha$-conversion} (renaming of bound names).
In case $w \in\n(P)$, $P\subs{w}{z}$ denotes $Q\subs{w}{z}$ for some process $Q$ obtained from $P$
by means of $\alpha$-conversion, such that $z$ does not occur within subexpressions $x(w).Q'$ or $(\nu w)Q'$ of $Q$.
\end{definition}

\begin{definition}\rm\label{df:reduction}
The \emph{reduction relation}, ${\longmapsto}\subseteq \T_\pi \times \T_\pi$, is generated by the
following rules:\vspace{-2ex}
\[\begin{array}{@{}cc@{}}
\displaystyle\frac{~}{\bar xz.P | x(y).Q \longmapsto P|Q\subs{z}{y}} &
\displaystyle\frac{P \longmapsto P'}{P|Q \longmapsto P'|Q} \\[3ex]
\displaystyle\frac{P \longmapsto P'}{(\nu y)P \longmapsto (\nu y)P'} &
\displaystyle\frac{Q \equred P \quad\! P \longmapsto P' \quad\! P' \equred Q'}{Q \longmapsto Q'}\;.
\end{array}\]
\end{definition}

\noindent
The asynchronous $\pi$-calculus, as introduced by Honda \& Tokoro in \cite{HT91} and by Boudol in
\cite{Bo92}, is the sublanguage $\rm a\pi$ of the fragment $\pi$ of the $\pi$-calculus presented above where all
subexpressions $\bar x z.P$ have the form $\bar x z.\textbf{0}$, and are written $\bar x z$. 
A characteristic of synchronous communication, as used in $\pims$,
is that sending a message synchronises with receiving it, so that
a process sending a message can only proceed after another party has
received it. In the asynchronous $\pi$-calculus this feature is dropped, as
it is not possible to specify any behaviour scheduled after a send action.

Boudol \cite{Bo92} defines an encoding $\fT_{\rm B}$ from $\pi$ to $\rm a\pi$ inductively as follows:
\[
\renewcommand{\sI}{\fT^{}_{\rm B}}
\begin{array}{r@{~~:=~~}l@{\qquad}l}
\sI(\nil)         & \nil \\
\sI(\bar{x}z.P)   & (u)(\bar{x}u | u(v).(\bar{v}z | \sI(P))) &\mbox{\small with $u,v\mathbin{\notin}\fn(P){\cup}\{x,z\}$} \\
\sI(x(y).P)       & x(u).(v)(\bar{u}v | v(y).\sI(P)) &\mbox{\small with $u,v\mathbin{\notin}\fn(P){\cup}\{x\}$} \\
\sI(P | Q)        & (\sI(P) | \sI(Q)) \\
\sI(!P)           & \ ! \sI(P) \\
\sI((x) P)        & (x) \sI(P)
\end{array}
\]
always choosing $u\mathbin{\neq} v$.
To sketch the underlying idea, suppose a $\pims$-process is able to perform a
communication, for example $\bar{x}z.P|x(y\hspace{-.6pt}).Q$.
In the asynchronous variant of the \piCal, there is no continuation process after an output operation.
Hence, a translation into the asynchronous {\piCal} has to reflect the communication on channel $x$ as well as the guarding role of $\bar x z$ for $P$ in the synchronous \piCal.
The idea of Boudol's encoding is to assign a guard to $P$ such that this process must receive an
{\em acknowledgement message} confirming the receipt of $z$.%
\footnote{As observed by a referee, the encodings $\fT_{\rm B}$ and $\fT_{\rm HT}$ do not satisfy this constraint:
the continuation process $P$ can proceed before $z$ is received. This issue could be alleviated by
enriching the protocol with another communication from $Q'$ to $P'$.}
We write the sender as
$P' \mathbin= (\bar{x}z|u(v).P)$
where $u,v \mathbin{\not\in} \fn(P)$.
Symmetrically, the receiver must send the acknowledgement, \ie
$Q' = x(y).(\bar{u}v|Q)$.
Unfortunately, this simple transformation is not applicable in every
case, because the protocol does not protect the channel $u$.
$u$ should be known to sender and receiver only, otherwise the
communication may be interrupted by the environment.
Therefore, we restrict the scope of $u$, and start by sending this
private channel to the receiver. The actual message $z$ is now sent in
a second stage, over a channel $v$, which is also made into a
private channel between the two processes.
The crucial observation is that in
$(u)(\bar{x}u|u(v).P^*)$,
the subprocess $P^*=\bar{v}z | P$ may only continue after
$\bar{x}u$ was accepted by some receiver, and this receiver has
  acknowledged this by transmitting another channel name $v$ on the
  private channel $u$.

The encoding $\fT_{\rm HT}$ of Honda \& Tokoro \cite{HT91} differs only in the
clauses for the input and output prefix:
\[
\renewcommand{\sI}{\fT^{}_{\rm HT}}
\begin{array}{r@{~~:=~~}l@{\qquad}l}
\sI(\bar{x}z.P)   & x(u).(\bar{u}z | \sI(P) ) &\mbox{\small$u\mathbin{\notin}\fn(P){\cup}\{x,z\}$} \\
\sI(x(y).P)       & (u)(\bar{x}u | u(y).\sI(P)) &\mbox{\small$u\mathbin{\notin}\fn(P){\cup}\{x\}$.} \\
\end{array}
\]
Unlike Boudol's translation, communication
takes place directly after synchronising along the private channel $u$.
The synchronisation occurs in the reverse direction, because sending
and receiving messages alternate, meaning that the sending process 
$\bar{x}z.Q$ is translated into a process that receives a message 
on channel $x$ and the receiving process $x(y).R$ is translated into a 
process passing a message on $x$.

\section{Valid Encodings According to Gorla}\label{sec:Gorla-validity}

\noindent
In \cite{Gorla10a} a \emph{process calculus} is given as a triple
$\fL\mathbin=(\mathcal{P},\longmapsto,\asymp)$, where
\begin{itemize}
\item $\mathcal{P}$ is the set of language terms (called
  \emph{processes}), built up from $k$-ary composition operators $\op$,\pagebreak[3]
\item $\longmapsto$ is a binary \emph{reduction} relation between
  processes,
\item $\asymp$ is a semantic equivalence on processes.
\end{itemize}
The operators themselves may be constructed from a set $\N$
of names. In the $\pi$-calculus, for instance, there is a unary
operator $\bar x y.\_$ for each pair of names
\mbox{$x,y\mathbin\in\N$}.
This way names occur in processes; the occurrences of names in
processes are distinguished in \emph{free} and \emph{bound} ones;
$\fn(\vec P)$ denotes the set of names occurring free in the $k$-tuple
of processes $\vec P=(P_1,\dots,P_k)\mathbin\in\mathcal{P}^k$.
A \emph{renaming} is a function $\sigma:\N\rightarrow\N$;
it extends componentwise to $k$-tuples of names.
If $P\mathbin\in\mathcal{P}$ and $\sigma$ is a renaming, then $P\sigma$
denotes the term $P$ in which each free occurrence of a name $x$ is
replaced by $\sigma(x)$, while renaming bound names to avoid name capture.

A $k$-ary $\fL$-\emph{context} $C[\__1,\dots,\__k]$ is a term
build by the composition operators of $\fL$ from \emph{holes}
$\__1,\dots,\__k$\,; the context is called \emph{univariate}
if each of these holes occurs exactly once in it.
If $C[\__1,\dots,\__k]$ is a $k$-ary $\fL$-\emph{context} and
$P_1,\dots,P_k \in \mathcal{P}$ then $C[P_1,\dots,P_k]$ denotes the
result of substituting $P_i$ for $\__i$ for each $i\mathbin=1,\dots,k$.

Let $\Longmapsto$ denote the reflexive-transitive closure of $\longmapsto$.
One writes $P\longmapsto^\omega$ if $P$ \emph{diverges}, that is, if
there are $P_i$ for $i\in\mbox{\bbb N}$ such that $P\mathbin= P_0$ and
$P_i\longmapsto P_{i+1}$ for all $i\mathbin\in\mbox{\bbb N}$.
Finally, write $P\longmapsto$ if $P\longmapsto Q$ for some term $Q$.

For the purpose of comparing the expressiveness of languages,
a constant $\surd$ is added to each of them~\cite{Gorla10a}.
A term $P$ in the upgraded language is said to \emph{report success},
written $P{\downarrow}$, if it has a \emph{top-level unguarded} occurrence
of $\surd$.\footnote{Gorla defines the latter concept only for languages
  that are equipped with a notion of \emph{structural congruence} $\equiv$ as well as a parallel
  composition $|$. In that case
  $P$ has a top-level unguarded occurrence of $\surd$ iff $P\equiv Q|\surd$, for some $Q$~\cite{Gorla10a}.
  Specialised to the $\pi$-calculus, a \emph{(top-level) unguarded} occurrence is one that not lies
  strictly within a subterm $\alpha.Q$, where $\alpha$ is $\tau$, $\bar xy$ or $x(z)$.
  For De Simone languages \cite{dS85}, even when not equipped with $\equiv$ and $|$, a suitable
  notion of an unguarded occurrence is defined in \cite{Va93copy}.}
Write $P{\Downarrow}$ if $P\Longmapsto P'$ for a process $P'$ with $P'{\downarrow}$.

\begin{definition}[{\cite{Gorla10a}}]\rm\label{df:encoding}
An \emph{encoding} of $\fL_{\rm s}\mathbin=(\mathcal{P}_{\rm s},\longmapsto_{\rm s},\asymp_{\rm s})$
into $\fL_{\rm t}\mathbin=(\mathcal{P}_{\rm t},\longmapsto_{\rm t},\asymp_{\rm t})$ is a pair
$(\fT,\varphi_{\subtrans})$ where
$\fT:\mathcal{P}_{\rm s}\rightarrow\mathcal{P}_{\rm t}$ is called \emph{translation}
and $\varphi_{\subtrans}:\N\rightarrow\N^k$ for some $k\mathbin\in\mbox{\bbb N}$
is called \emph{renaming policy} and is such that for $u\neq v$ the
$k$-tuples $\varphi_{\subtrans}(u)$ and $\varphi_{\subtrans}(v)$ have no name in common.
\end{definition}
\noindent
The terms of the source and target languages $\fL_{\rm s}$ and
$\fL_{\rm t}$ are often called $S$ and $\pT$, respectively.

\begin{definition}[{\cite{Gorla10a}}]\rm\label{df:valid}
An encoding is \emph{valid} if it satisfies the following five criteria.
\begin{enumerate}
\item \emph{Compositionality:} for every $k$-ary operator $\op$ of
  $\fL_{\rm s}$ and for every set of names $N\subseteq\N$,
  there exists a univariate $k$-ary context $C_\op^N[\__1,\dots,\__k]$ such that
  $$\trans{\op(S_1,\ldots,S_k)} = C_\op^N[\trans{S_1},\ldots,\trans{S_k}]$$
  for all $S_1,\ldots,S_k\in\mathcal{P}_{\rm s}$ with $\fn(S_1,\dots,S_n)=N$.
\item \emph{Name invariance:} for every $S\mathbin\in\mathcal{P}_{\rm s}$ and
  $\sigma:\N\rightarrow\N$
  \[\begin{array}{ccc@{\quad}l}
    \trans{S\sigma} & = & \trans{S}\sigma' & \mbox{if $\sigma$ is injective}\\
    \trans{S\sigma} & \asymp_{\rm t} & \trans{S}\sigma' & \mbox{otherwise}\\
    \end{array}\]
  with $\sigma'$ such that
  $\varphi_{\subtrans}(\sigma(a))\mathbin=\sigma'(\varphi_{\subtrans}(a))$
  for all $a\mathbin\in\N\!$.
\item \emph{Operational correspondence:}\\
  \begin{tabular}{@{}ll}
   \emph{Completeness} & if $S \Longmapsto_{\rm s} S'$ then $\trans{S}\Longmapsto_{\rm t}\asymp_{\rm t} \trans{S'}$\\
   \emph{Soundness} & if $\trans{S} \Longmapsto_{\rm t} \pT$ then $\exists S'\!: S\Longmapsto_{\rm s} S'$ and $\pT\Longmapsto_{\rm t}\asymp_{\rm t} \trans{S'}$.
  \end{tabular}
\item \emph{Divergence reflection:} 
  if $\trans{S} \longmapsto_{\rm t}^\omega$ then $S \longmapsto_{\rm s}^\omega$.
\item \emph{Success sensitiveness:}
  $S {\Downarrow}$ iff $\trans{S}{\Downarrow}$.\\
  For this purpose $\trans{\cdot}$ is extended to deal with the added constant $\surd$ by taking $\trans{\surd}=\surd$.
\end{enumerate}
\end{definition}
The above treatment of success sensitiveness differs slightly from the one of Gorla \cite{Gorla10a}.
Gorla requires $\surd$ to be a constant of any two languages whose expressiveness is
compared. Strictly speaking, this does not allow his framework to be applied to the encodings
$\fT_{\rm B}$ and $\fTHT$, as these deal with languages not featuring $\surd$.  Here, following
\cite{vG18a}, we simply allow $\surd$ to be added, which is in line with the way Gorla's framework
has been used \cite{Gorla10b,LPSS10,PSN11,PN12,PNG13,GW14,EPTCS160.4,EPTCS189.9,GWL16}.
A consequence of this decision is that one has to specify how $\surd$ is translated---see the last
sentence of Definition~\ref{df:valid}---as the addition of $\surd$ to both languages happens after a
translation is proposed. This differs from \cite{Gorla10a}, where it is explicitly allowed to take
$\trans{\surd} \neq \surd$.

In \cite{vG18a} it is established that the encodings $\fT_{\rm B}$ and $\fTHT$, reviewed in \sect{pi2api},
are valid according to Gorla \cite{Gorla10a}; that is, both encodings enjoy the five correctness criteria above.
Here, the semantic equivalences $\asymp_{\rm s}$ and $\asymp_{\rm t}$ that Gorla assumes to exist on the source
and target languages, but were not specified in \sect{pi2api}, can chosen to be the identity, thus
obtaining the strongest possible instantiation of Gorla's criteria.
Moreover, the renaming policy required by Gorla as part of an encoding can be chosen to be the identity,
taking $k\mathbin=1$ in \df{encoding}.
Trivially, $\fT_{\rm B}$ and $\fT_{\rm HT}$ also meet Palamidessi's criterion that
parallel composition is translated homomorphically, and thus also the criterion on preservation of
distributability from \cite{PNG13}.

\section{Compositionality}\label{sec:compositionality}

Compositionality demands that for every $k$-ary operator $\op$ of the source language
there is a $k$-ary context $C_\op[\__1,\dots,\__k]$ in the target such that
  $$\trans{\op(S_1,\ldots,S_k)} = C_\op[\trans{S_1},\ldots,\trans{S_k}]$$
for all $S_1,\ldots,S_k\in\mathcal{P}_{\rm s}$ \cite{Bo85}. Gorla \cite{Gorla10a} strengthens this
requirement by the additional requirement that the context $C_\op$ should be univariate; at the same
time he weakens the requirement by allowing the required context $C_\op$ to 
depend on the set of names $N$ that occur free in the arguments $S_1,\ldots,S_k$.
The application to the encodings $\fT_{\rm B}$ and $\fT_{\rm HT}$ shows that we cannot simply
strengthen the criterion of compositionality by dropping the dependence on $N$. For then the present
encodings would fail to be compositional. Namely, the context $C_{\bar xz.\_}$ depends on the choice of two
names $u$ and $v$, and the choice of these names depends on $N=\fn(S_1)$, where $S_1$ is the only
argument of output prefixing. That the choice of  $C_{\bar xz.\_}$ also depends on $x$ and $z$ is unproblematic.

In \cite{vG18b} a form of compositionality is proposed where $C_\op$ does not depend on $N$,
but the main requirement is weakened to
  $$\trans{\op(S_1,\ldots,S_k)} \stackrel\alpha= C_\op[\trans{S_1},\ldots,\trans{S_k}].$$
Here $\stackrel\alpha=$ denotes equivalence up to \emph{$\alpha$-conversion}, renaming of bound
names and variables, for the $\pi$-calculus corresponding with rules $\scriptstyle{(8)}$ and $\scriptstyle{(9)}$ of
structural congruence. This suffices to rescue the current encodings, for up to $\alpha$-conversion
$u$ and $v$ can always be chosen outside $N$. It is an open question whether
there are examples of intuitively valid encodings that essentially need the dependence of $N$ allowed
by \cite{Gorla10a}, i.e., where $C_\op^{N_{\rm s}}$ and $C_\op^{N_{\rm t}}$ differ by more than $\alpha$-conversion.

Another method of dealing with the fresh names $u$ and $v$ that are used in the encodings $\fT_{\rm B}$ and $\fT_{\rm HT}$,
proposed in \cite{vG12}, is to equip the target language with two fresh names that do not occur
in the set of names available for the source language. Making the dependence on the choice of set
$\N$ of names explicit, this method calls $\pims$ expressible into $\pima$ if for each $\N$ there
exists an $\N'$ such that there is a valid encoding of $\pi(\N)$ into $\pima(\N')$.
By this definition, the encodings $\fT_{\rm B}$ and $\fT_{\rm HT}$ even satisfy the default
definition of compositionality, and its strengthening obtained by insisting the contexts $C_\op$ to be univariate.

\section{Operational Correspondence}\label{sec:operational correspondence}

Operational completeness (one half of operational correspondence) was formulated by Nestmann \&
Pierce \cite{NestmannP00} as
\begin{center}
\mbox{}\hfill
$S \longmapsto_{\rm s} S'$ then $\trans{S}\Longmapsto_{\rm t} \trans{S'}$.\hfill$(\mathfrak{C})$
\end{center}
It makes no difference whether the antecedent of this implication is rephrased as $S \Longmapsto_{\rm s} S'$,
as done by Gorla. Gorla moreover weakens the criterion to
\begin{center}
\mbox{}\hfill
$S \Longmapsto_{\rm s} S'$ then $\trans{S}\Longmapsto_{\rm t}\asymp_{\rm t} \trans{S'}$.\hfill$(\mathfrak{C}')$
\end{center}
This makes the criterion applicable to many more encodings. In the case of $\fT_{\rm B}$ and $\fT_{\rm HT}$,
\cite{vG18a} shows that these encodings not only satisfy $(\mathfrak{C}')$, but even $(\mathfrak{C})$.

Operational soundness also stems from Nestmann \& Pierce \cite{NestmannP00}, who proposed two forms of it:
\begin{center}
\mbox{}\hfill
if $\trans{S} \longmapsto_{\rm t} \pT$ then $\exists S'\!:$ $S\longmapsto_{\rm s} S'$ and $\pT\asymp_{\rm t} \trans{S'}$.~~~\hfill$(\mathfrak{I})$\\
\mbox{}\hfill
if $\trans{S} \Longmapsto_{\rm t} \pT$ then $\exists S'\!:$ $S\Longmapsto_{\rm s} S'$ and $\pT\Longmapsto_{\rm t}\trans{S'}$.\hfill$(\mathfrak{S})$
\end{center}
The former is meant for ``\emph{prompt} encodings, i.e., those where initial steps of literal
translations are committing'' \cite{NestmannP00}, whereas the latter apply to ``nonprompt encodings'',
that ``allow administrative (or book-keeping) steps to precede a committing step''.
The version of Gorla is the common weakening of $(\mathfrak{I})$ and $(\mathfrak{S})$:
\begin{center}
\mbox{~~\,}\hfill
if $\trans{S} \mathbin{\Longmapsto_{\rm t}} \pT$ then $\exists S'\!:\!S\mathbin{\Longmapsto_{\rm s}} S'$ and $\pT\Longmapsto_{\rm t}\asymp_{\rm t}\trans{S'}$.\hfill$(\mathfrak{G})$
\end{center}
It thus applies to prompt as well as nonprompt encodings.
The encodings $\fT_{\rm B}$ and $\fT_{\rm HT}$ are nonprompt, and accordingly do not meet $(\mathfrak{I})$.
In \cite{vG18a} it was shown that they not only satisfy $(\mathfrak{G})$, but even $(\mathfrak{S})$.

An interesting intermediate form between $\mathfrak{J}$ and $\mathfrak{G}$ is
\begin{center}
\mbox{}\hfill
if $\trans{S} \Longmapsto_{\rm t} \pT$ then $\exists S'\!:$ $S\Longmapsto_{\rm s} S'$ and $\pT\asymp_{\rm t} \trans{S'}$.\hfill~~$(\mathfrak{W})$
\end{center}
Whereas $(\mathfrak{G})$ weakens $(\mathfrak{J})$ in two ways, $(\mathfrak{W})$ weakens
$(\mathfrak{J})$ in only one of these ways.
Moreover, $(\mathfrak{W})$ is the natural counterpart of $(\mathfrak{C}')$.
In \cite{vG18a} the open question was raised whether $\fT_{\rm B}$ and $\fT_{\rm HT}$ satisfy
$(\mathfrak{W})$, for a reasonable choice of $\asymp_{\rm t}$.
(An unreasonable choice, such as the universal relation, tells us nothing.)
As pointed out in \cite{vG18a}, they do not when taking $\asymp_{\rm t}$ to be the identity relation, or
structural congruence.

The present paper answers this question affirmatively, taking  $\asymp_{\rm t}$
to be weak barbed bisimilarity. A proof will follow in \sect{versus}.

\section{Barb Sensitiveness}\label{sec:barbs}

Gorla's success predicate is one of the possible ways to provide source and target languages with
a set of \emph{barbs} $\Omega$, each being a unary predicate on processes. For $\omega \in \Omega$,
write $P{\downarrow_\omega}$ if process $P$ has the barb $\omega$, and 
$P{\Downarrow_\omega}$ if $P\Longmapsto P'$ for a process $P'$ with $P'{\downarrow_\omega}$.
In Gorla's case, $\Omega=\{\surd\}$, and $P$ has the barb $\surd$ iff $P$ has a top-level unguarded
occurrence of $\surd$.
The standard criterion of barb sensitiveness is then
$S{\Downarrow_\omega} \Leftrightarrow \trans{S}{\Downarrow_\omega}$ for all $\omega\in\Omega$.

A traditional choice of barb in the $\pi$-calculus is to take $\Omega=\{x,\bar x \mid x\mathbin\in\N\}$,
writing $P{\downarrow_x}$, resp.\ $P{\downarrow_{\bar x}}$, when $P$ has an unguarded occurrence of
a subterm $x(z).R$, resp.\ $\bar xy.R$, that lies not in the scope of a restriction operator $\nu(x)$ \cite{milner:poly,SW01book}.
This makes a barb a predicate that tells weather a process can read or write over a given channel.
Boudol's encoding keeps the original channel names of a sending or receiving process invariant.
Hence, a translated term does exhibit the same barbs as the source term.
\begin{lemma}\label{lem:barbs}\rm
Let $P \in \bbT_{\pims}$ and $a \in \{ x, \bar x \,|\, x \in \mathcal N \}$.
Then $P \scomm a$ iff $\fT_{\rm B}(P) \scomm a$.
\end{lemma}
\begin{proof}
\renewcommand{\sI}{\fT_{\rm B}}
With structural induction on $P$.
\leftmargini 20pt
\begin{itemize}
\item $\bm{0}$ and $\sI(\bm{0})$ have the same strong barbs, namely none.%
\item $\bar xz.P$ and $\sI(\bar xz.P)$ both have only the strong barb $\bar x$.%
\item $x(y).P$ and $\sI(x(y).P)$ both have only strong barb $x$.%
\item The strong barbs of $P|Q$ are the union of the ones of $P$ and $Q$.
  Using this, the case $P|Q$ follows by induction.%
\item The strong barbs of $!P$ are the ones of $P$.
  Using this, the case $!P$ follows by induction.
\item The strong barbs of $(x)P$ are ones of $P$ except $x$ and $\bar x$.
  Using this, the case $(x)P$ follows by induction.
\qed
\end{itemize}
\end{proof}
It follows that $\fT_{\rm B}$ meets the validity criterion of barb sensitiveness.

The philosophy behind the asynchronous $\pi$-calculus entails that input actions $x(z)$ are not
directly observable (while output actions can be observed by means of a matching input of
the observer). This leads to semantic identifications like $\nil = x(y).\bar xy$, for in
both cases the environment may observe $\bar x z$ only if it supplied $\bar x z$ itself first.
Yet, these processes differ on their input barbs ($\downarrow_x$). For this reason, in
$\rm a\pi$ normally only output barbs ${\downarrow_{\bar x}}$ are considered \cite{SW01book}.
Boudol's encoding satisfies the
criterion of output barb sensitiveness (and in fact also input barb sensitiveness).
However, the encoding of Honda \& Tokoro does not, as it swaps input and output barbs.
As such, it is an excellent example of the benefit of the external barb $\surd$ employed in Gorla's
notion of success sensitiveness.

To obtain a weaker form of barb sensitiveness such that also $\fT_{\rm HT}$ becomes barb sensitive,
we introduce \emph{channel barbs} $x\in\N$. A process is said to have the channel barb $x$ iff it
either has the barb $\bar x$ or $x$. We write $P \!\ascomm x$ when $P$ has the channel barb $x$,
and $P\acomm x$ when a $P'$ exists with $P \Longmapsto P'$ and $P' \!\ascomm x$.

\begin{definition}\rm
  An encoding $\fT$ is \emph{channel barb sensitive} if
  $S{\Downarrow_\omega} \Leftrightarrow \trans{S}{\Downarrow_\omega}$ for all $\omega\in\Omega$.
\end{definition}
This is a weaker criterion than barb sensitiveness, so $\fT_{\rm B}$ is surely channel barb sensitive.
It is easy to see that Honda \& Tokoro's encoding $\fT_{\rm HT}$, although not barb sensitive, is
channel barb sensitive.

\section{Validity up to a Semantic Equivalence}\label{sec:validity upto}

This section deals with the original validity criterion from Boudol \cite{Bo85}, as generalised in \cite{vG12}.
Following \cite{vG12} we call a compositional encoding $\fT$ \emph{valid up to} a semantic equivalence
${\sim} \subseteq \T \times \T$, where $\T\supseteq \T_{\rm s} \cup \T_{\rm t}$, iff $\trans{P} \sim P$ for all $P\in\T_{\rm s}$.
A given encoding may be valid up to a coarse equivalence, and invalid up to a finer one.
The equivalence for which it is valid is then a measure of the quality of the encoding.

Below, we will evaluate the encodings $\fT_{\rm B}$ and $\fT_{\rm HT}$ under a number of semantic
equivalences found in the literature. Since these encodings translate a single transition in the
source language by a small protocol involving two or three transitions in the target language, they
surely will not be valid under \emph{strong} equivalences, demanding step-for-step matching of
source transitions by target transitions. Hence we only look at \emph{weak} equivalences.

First we consider equivalences that match transition labels, such as early weak bisimilarity.
The encodings $\fT_{\rm B}$ and $\fT_{\rm HT}$ are not valid under such equivalences.
Then we show that Boudol's encoding $\fT_{\rm B}$ is valid under weak barbed bisimilarity, and thus certainly under
its asynchronous version; however, it is not valid under asynchronous weak bisimilarity.
The encoding $\fTHT$ of Honda \& Tokoro is not valid under any of these equivalences,
but we introduce a new equivalence under which it is valid: a version of weak barbed bisimilarity
that drops the distinction between input and output barbs.

\subsection{A Labelled Transition Semantics of \texorpdfstring{$\pi$}{the pi-calculus}}\label{sec:processcalculi}

We first present a labelled transition semantics of the (a)synchronous $\pi$-calculus,
to facilitate the definition of semantic equivalences on these languages.
Its labels are drawn from
a set of actions $\Act := \{ \bar x y, x(y), \bar x(y) \,|\, x, y \in \mathcal N \} \cup \{ \tau \}$.
We define free and bound names on transition labels:\vspace{-4ex}

\begin{align*}
\fn(\tau)       &= \emptyset &\bn(\tau)       &= \emptyset \\[-3pt]
\fn(\bar{x}z)   &= \{x,z\}   &\bn(\bar{x}z)   &= \emptyset \\[-3pt]
\fn(x(y))       &= \{x\}     &\bn(x(y))       &= \{y\} \\[-3pt]
\fn(\bar{x}(y)) &= \{x\}    &\bn(\bar{x}(y)) &= \{y\}\;.
\end{align*}
For $\alpha\in\Act$ we define $\n(\alpha) := \bn(\alpha) \cup \fn(\alpha)$.

\begin{definition}\rm\label{df:syncltssemantics}
The {\em labelled transition relation of $\pims$} is the
smallest relation $\mathord{\longrightarrow} \subseteq \bbT_{\pims} \times \Act \times \bbT_{\pims}$,
satisfying the rules of \tab{pisyncsemantics}.
\end{definition}

{\renewcommand*\arraystretch{3}%
\newcommand{\inference}[3][]{{\scriptsize #1}$\displaystyle\frac{#2}{#3}$}
\begin{table}[t]
\centering
\vspace{-2ex}%
\begin{center}
\begin{tabular}{@{}cc@{}}
\inference[(OUTPUT-ACT)~]
{~}
{\bar{x}z.P\stackrel{\bar{x}z}{\longrightarrow}P}
&
\inference[(INPUT-ACT)~]
{w\not\in\fn((y)P)}
{x(y).P\stackrel{x(w)}{\longrightarrow}P\subs{w}{y}}
\\

\inference[(PAR)~]
{P \stackrel{\alpha}{\longrightarrow} P' \qquad \bn(a) \cap \fn(Q) = \emptyset}
{P|Q \stackrel{\alpha}{\longrightarrow} P'|Q}
&

\inference[(COM)~]
{P\stackrel{\bar{x}z}{\longrightarrow}P' \qquad Q\stackrel{x(y)}{\longrightarrow} Q'}
{P|Q \stackrel{\tau}{\longrightarrow}P'|Q'\subs{z}{y}}
\\
\inference[(CLOSE)~]
{P\stackrel{\bar{x}(w)}{\longrightarrow}P' \qquad Q\stackrel{x(w)}{\longrightarrow}Q'}
{P|Q\stackrel{\tau}{\longrightarrow}(w)(P'|Q')}
&
\inference[(RES)~]
{P \stackrel{\alpha}{\longrightarrow}P' \qquad y \not\in \n(a)}
{(y)P\stackrel{\alpha}{\longrightarrow}(y)P'}
\\
\inference[(OPEN)~]
{P \stackrel{\bar{x}y}{\longrightarrow} P' \qquad y \neq x \qquad w \not\in \fn((y)P')}
{(y)P \stackrel{\bar{x}(w)}{\longrightarrow}P'\subs{w}{y}}
&
\inference[(REP-ACT)~]
{P \stackrel{\alpha}{\longrightarrow} P'}
{!P \stackrel{\alpha}{\longrightarrow} P'| !P}
\\
\inference[(REP-COMM)~]
{P \stackrel{\bar x z}{\longrightarrow} P' \qquad P \stackrel{x(y)}{\longrightarrow} P''}
{!P \stackrel{\tau}{\longrightarrow} ( P' | P''\subs{z}{y} ) | !P}
&
\inference[(REP-CLOSE)~]
{P \stackrel{\bar x(w)}{\longrightarrow} P' \qquad P \stackrel{x(w)}{\longrightarrow} P''}
{!P \stackrel{\tau}{\longrightarrow} ( (w)( P'|P'' ) ) | !P}
\end{tabular}
\end{center}
\caption{SOS rules for the synchronous mini-\piCal.
PAR, COM and CLOSE also have symmetric rules.}
\label{tab:pisyncsemantics}\vspace{-2ex}
\end{table}}

\noindent
The $\tau$-transitions in the labelled transition semantics play the same role as the reductions in
the reduction semantics: they present actual behaviour of the represented system.
The transitions with a label different from $\tau$ merely represent potential behaviour:
a transition $x(y)$ for instance represents the potential of the system to receive a value on
channel $x$, but this potential will only be realised in the presence of a parallel component that
sends a value on channel $x$. Likewise, an output action $\bar x z$ or $\bar x (y)$ can be realised
only in communication with an input action $x(y)$.

The following results show (1) that the labelled transition relations are invariant under structural congruence ($\equred$), and (2)
that the closure under structural congruence of the
labelled transition relation restricted to $\tau$-steps coincides with
the reduction relation --- (2) stems from Milner~\cite{Mi92}.
\begin{lemma}[{\rm Harmony Lemma \cite[Lemma 1.4.15]{SW01book}}]\label{harmony}\rm
\begin{enumerate}
\item If $P \stackrel{\alpha}{\longrightarrow} P'$ and $P\equred Q$ then $\exists Q'. Q \stackrel{\alpha}{\longrightarrow} Q' \equred P'$
\item $P \longmapsto P'$ iff $\exists P''. P \stackrel{\tau}{\longrightarrow} P'' \equred P'$.
\end{enumerate}
\end{lemma}
The barbs defined in \sect{barbs} can be characterised in terms of the labelled transition relation
as follows:
\begin{remark}
A process $P$ has a {\em strong barb on $x\in\N$}, $P \sbarb x$, iff there
is a $P'$ with \plat{$P \stackrel{x(y)}{\longrightarrow} P'$} for some $y\mathbin\in\N$.
It has a {\em strong barb on $\bar x$}, $P \sbarb {\bar x}$, iff there\vspace{1pt}
is a $P'$ with \plat{$P \stackrel{\bar x z}{\longrightarrow} P'$} or
\plat{$P \stackrel{\bar x(z)}{\longrightarrow} P'$} for some $z\mathbin\in\N$.
A process $P$ has a \mbox{\em weak barb on $a$}
($a \in \{ x, \bar x \,|\, x \in \mathcal N \}$), $P \ocomm a$, iff
there is a $P'$ such that \plat{$P \mathrel{\raisebox{0pt}[8pt][0pt]{$\stackrel{\tau}{\longrightarrow}$}^*} P'$} and
$P' \sbarb a$.

A process $P$ has a {\em channel barb on $x$}, $P \!\ascomm x$,
iff it can perform an action on channel $x$, \ie 
iff \plat{$P \stackrel{\alpha}{\longrightarrow} P'$},
for some $P'$, where $\alpha$ has the form $\bar{x}y$, $\bar{x}(y)$ or $x(y)$.
Moreover, $P\acomm x$ iff a $P'$ exists with \plat{$P \mathrel{\raisebox{0pt}[8pt][0pt]{$\stackrel{\tau}{\longrightarrow}$}^*} P'$} and
and $P' \!\ascomm x$.
\end{remark}

\subsection{Comparing Transition Labels: Early and Late Weak Bisimilarity}

As they make use of intermediate steps (namely the acknowledgement protocol), we must fail proving
the validity of the encodings $\fT_{\rm B}$ or $\fTHT$ up to semantics based on transition labels,
\eg {\em early weak bisimilarity}~\cite{SW01book}.

\begin{definition}\rm\label{df:WB}
A symmetric binary relation {\mrel} on $\pi$-processes $P,Q$ is a {\em early weak
bisimulation} iff $P \mrel Q$ implies
\begin{enumerate}
\item if $P \stackrel{\tau}{\longrightarrow} P'$ then a $Q'$ exists with
  $Q \transtau Q'$ and $P'\mrel Q'$,
\item if $P \stackrel{\alpha}{\longrightarrow} P'$ where $\alpha\mathbin= \bar{x}z$ or
  $\bar{x}(y)$ with $y \mathbin{\notin} \n(P)\cup\n(Q)$,
  then a $Q'$ exists with
  \plat{$Q \transtau\stackrel{\alpha}{\longrightarrow}\transtau Q'$} and $P' \mrel Q'$,
\item if \plat{$P \stackrel{x(y)}{\longrightarrow} P'$} with $y \notin \n(P)\cup\n(Q)$ then for all $w$ a $Q'$
  exists satisfying \plat{$Q \transtau\stackrel{x(y)}{\longrightarrow} \transtau Q'$}
  and  $P'\subs{w}{y} \mrel Q'\subs{w}{y}$.
\end{enumerate}
We denote the largest early weak bisimulation by $\wesim$.
\end{definition}

\noindent
Here $y \notin \n(P)\cup\n(Q)$ merely ensures the usage of fresh names.
A \emph{late} weak bisimulation is obtained by requiring in Clause 3 above that the choice of
$Q'$ is independent of $w$; this gives rise to a slightly finer equivalence relation.

\begin{observation}\label{obs:boudolwesim}\rm
$\fT_{\rm B}$ is not valid up to $\wesim$.
\end{observation}

\begin{proof}
Let $P = \bar{x}z.\bm{0}$ and $\fT_{\rm B}(P) = (u)(\bar{x}u|u(v).(\bar{v}z|\bm{0}))$.
We present the relevant parts of the labelled transition semantics:
\begin{center}
\begin{tikzpicture}[shorten >=1 pt,>=latex',scale=1.0]
  \node(b0) at(0,3){$\bar{x}z.\bm{0}$};
  \node(b1) at(0,1.5){$\bm{0}$};
    \draw[->](b0)--node[mode,right]{\bar{x}z}(b1);
  \node(b0) at(4.5,4.5){$(u)(\bar{x}u|u(v).(\bar{v}z|\bm{0}))$};
  \node(b1) at(4.5,3){$(u)(\bm{0}|c(v).(\bar{v}z|\bm{0}))$};
  \node(b2) at(4.5,1.5){$(u)(\bm{0}|(\bar{d}z|\bm{0}))$};
  \node(b3) at(4.5,0){$(u)(\bm{0}|(\bm{0}|\bm{0}))$};
    \draw[->](b0)--node[mode,right]{\bar{x}(c)}(b1);
    \draw[->](b1)--node[mode,right]{c(d)}(b2);
    \draw[->](b2)--node[mode,right]{\bar{d}z}(b3);
\end{tikzpicture}
\end{center}
Here, the translated term may perform an input transition \plat{$\stackrel{c(d)}{\longrightarrow}$} the source term is not capable of.
Hence, the processes are not equivalent up to $\wesim$.
\qed
\end{proof}
Since late weak bisimilarity is even finer (more discriminating) than $\wesim$, the encoding
$\fT_{\rm B}$ is certainly not valid up to late weak bisimilarity. A similar argument shows that neither
$\fTHT$ is valid up to early or late weak bisimilarity.

\subsection{Weak Barbed Bisimilarity}\label{sec:wbb}

A weaker approach does not compare all the transitions with
  visible labels, for these are merely \emph{potential} transitions,
  that can occur only in certain contexts. Instead it just compares
  internal transitions, together with the information whether a state
  has the potential to perform an input or output over a certain channel:
the barbs of \sect{barbs}.
Combining the notion of barbs
with the transfer property of classical bisimulation for internal actions only
yields {\em weak barbed bisimilarity} \cite{milner:poly}.
Here, two related processes simulate each other's internal transitions and
furthermore have the same weak barbs.

\begin{definition}\rm\label{df:WBB}
A symmetric relation {\mrel} on $\T_\pi$ is a {\em weak barbed bisimulation}
iff $P \mrel Q$ implies
\begin{enumerate}
\item if $P \scomm a$ with $a \in \{ x, \bar x \,|\, x \in \mathcal N \}$ then $Q \ocomm a$ and
\item if $P \stackrel{\tau}{\longrightarrow} P'\hspace{-1pt}$ then a $Q'$ exists with
  $Q \transtau Q'\hspace{-1pt}$ and $P'\mathop{\mrel} Q'\!$.%
\end{enumerate}
The largest weak barbed bisimulation is denoted by $\wbb$, or $\approx_{\mathrm{WBB}}$.
\end{definition}
\noindent
By Lemma~\ref{harmony} this definition can equivalently be stated with $\longmapsto$ in the role of $\stackrel{\tau}{\longrightarrow}$.
One of the main results of this paper is that Boudol's encoding is valid up to $\wbb$.
The proof of this result is given in the \hyperlink{appendix}{appendix}.

\subsection{Asynchronous Weak Barbed Bisimilarity}

In {\em asynchronous weak barbed bisimulation} \cite{acs:asynchronous},
only the names of output channels are observed.
Input barbs are ignored here, as it is assumed that an environment
is able to observe output messages, but not (missing) inputs.

\begin{definition}\rm\label{df:awbb}
A symmetric relation $S$ on $\T_\pi$ is an {\em asynchronous weak barbed bisimulation}
iff $P\mrel Q$ implies
\begin{enumerate}
\item if $P \scomm {\bar{x}}$, then $Q \wcomm {\bar{x}}$, and
\item if $P \stackrel{\tau}{\longrightarrow} P'\hspace{-1pt}$ then a $Q'\hspace{-1pt}$ exists with
  $Q \mathop{\stackrel{\tau}{\longrightarrow}^{\!*}}\! Q'\hspace{-1pt}$ and $P'\mathop{\mrel} Q'\!$.%
\end{enumerate}
The largest asynchronous weak barbed bisimulation is denoted by $\wbbisim$.
\end{definition}

\noindent
Since $\wbbisim$ is a coarser equivalence than $\wbb$, we obtain:

\begin{corollary}\rm
Boudol's encoding is valid up to $\wbbisim$.
\end{corollary}
In \cite{QW00}, a polyadic version of Boudol's encoding was assumed to be valid up to $\wbbisim$; see
Lemma 17. However, no proof was provided.

\subsection{Weak Asynchronous Bisimilarity}

\noindent
We now know that Boudol's translation is valid up to $\wbbisim$, but
not up to $\wesim$.
A natural step is to narrow down this gap by considering equivalences in between.
The most prominent semantic equivalence for the asynchronous {\piCal} is
weak asynchronous bisimilarity, proposed by Amadio et al.\ \cite{acs:asynchronous}.

A first strengthening of the requirements for $\wbbisim$ is obtained by considering not only
output channels but also the messages sent along them.

\begin{definition}[\cite{acs:asynchronous}]\rm
A symmetric relation {\mrel} on $\T_\pi$ is a {\em weak $o\tau$-bisimulation}
if {\mrel} meets Clauses 1 and 2 (but not necessarily 3) from \df{WB}.
The largest weak $o\tau$-bisimulation is denoted by $\wotau$.
\end{definition}

\noindent
Amadio et al.\ strengthen this equivalence by adding a further constraint
for input transitions.

\begin{definition}[\cite{acs:asynchronous}]\rm
A relation {\mrel} is a {\em weak asynchronous bisimulation}\vspace{1pt} iff {\mrel} is a
weak $o\tau$-bisimulation such that \mbox{$P\mrel Q$} and
\plat{$P \stackrel{\tau}{\longrightarrow}^* \stackrel{x(y)}{\longrightarrow}
\stackrel{\tau}{\longrightarrow}^* P'$} implies
\begin{itemize}
\item either a $Q'$ exists satisfying a condition akin to Clause 3 of
 \df{WB},
\item or a $Q'$ exists such that $Q \stackrel{\tau}{\longrightarrow}^* Q'$ and
$P'\mrel (Q'|\bar{x}y)$.
\end{itemize}
The largest weak asynchronous bisimulation is denoted by $\wab$.
\end{definition}

\begin{observation}\rm
Boudol's translation ${\fT_{\rm B}} : \bbT_{\pims} \rightarrow \bbT_{\pima}$ is
not valid up to $\wotau$, and thus not up to $\wab$.
\end{observation}

\begin{proof}
Consider the proof of \obs{boudolwesim}.
$\bar{x}z.\bm{0}$ sends a free name along $x$ while $(u)(\bar{x}u|u(v).(\bar{v}z|\bm{0}))$ sends a bound name along the same channel.
Since $\wotau$ differentiates between free and bound names, the transition systems of $\bar{x}z$ and its translation are not $\wotau$-equivalent.
\qed
\end{proof}

\subsection{Weak Channel Bisimilarity}
\renewcommand{\sI}{\fT_{\rm HT}}

From the equivalences considered, weak barbed bisimilarity, $\wbb$, is the finest one that supports
the validity of Boudol's translation. However, it does not validate Honda and Tokoro's translation.

\begin{observation}\rm
Honda and Tokoro's translation $\sI$ is not valid up to $\wbbisim$, and thus not up to $\wbb$,
$\wotau$, or $\wab$.
\end{observation}

\begin{proof}
Let $P = \bar{x}z.\bm{0}$.
Then $P \scomm {\bar{x}}$.
The translation is $\sI(P) = x(u).(\bar{u}z|\bm{0})$ and
$\sI(P) {\not\!\wcomm {\bar{x}}}$.
\qed
\end{proof}

\noindent
To address this problem we introduce an equivalence even weaker than $\wbb$, which does not distinguish between input and output channels.

\begin{definition}\rm\label{df:channel bis}
A symmetric relation {\mrel} on $\T_\pi$ is a {\em weak channel bisimulation}
if $P\mrel Q$ implies
\begin{enumerate}
\item if $P \ascomm x$ then $Q \acomm x$ and
\item if $P \stackrel{\tau}{\longrightarrow} P'$ then a $Q'$ exists with
  $Q \transtau Q'$ and $P'\mathop{\mrel} Q'\!$.%
\end{enumerate}
The largest weak channel bisimulation is denoted $\wlbsim$.
\end{definition}

\begin{theorem}\rm
Honda and Tokoro's encoding $\sI$
is valid up to $\wlbsim$.
\end{theorem}

\noindent
The proof is similar to the one of \thm{boudolwrequ}. 
Here we use that Lemmas~\ref{lem5} and ~\ref{lem6} also apply to $\sI$ \cite{vG18a}
and \lem{barbs} now holds with $\ascomm{x}$ in the role of $\scomm{a}$.

Since $\wlbsim$ is a coarser equivalence than $\wbb$, we  also obtain
that Boudol's translation is valid up to $\wlbsim$.

\subsection{Overview}

\noindent
We thus obtain the following hierarchy of equivalence relations on {\piCal} processes (\cf
Fig.~\ref{fig:semeq_hier}),
with the vertical lines indicating the realm of validity of $\fTHT$ and $\fT_{\rm B}$, respectively.\vspace{-2ex}
\begin{figure}[ht]
\centering
\begin{tikzpicture}[font=\small]

\node[circle,inner sep=.2em,fill=black,label=below:EWB] (ewb) at (12em,0) {};
\node[circle,inner sep=.2em,fill=black,label=above:WAB] (wab) at (8em,3em) {};
\node[rotate=-35] (wab2ewb) at (10em,1.5em) {$\supset$};

\node[circle,inner sep=.2em,fill=black,label=above:Wo$\tau$] (wot) at (4em,3em) {};
\node (wot2wab) at (6em,3em) {$\supset$};

\node[circle,inner sep=.2em,fill=black,label=below:WBB] (wbb) at (0,0) {};
\node (wbb2ewb) at (6em,0) {$\supset$};

\node (awbb2wot) at (1em,3em) {$\supset$};

\node[circle,inner sep=.2em,fill=black,label=above:AWBB] (awbb) at (-4em,3em) {};
\node[rotate=-35] (awbb2wbb) at (-2em,1.5em) {$\supset$};

\node[circle,inner sep=.2em,fill=black,label=below:WCB] (wcb) at (-8em,0) {};
\node (wcb2wbb) at (-4em,0) {$\supset$};

\draw[semithick] (wab) to (wab2ewb) to (ewb);
\draw[semithick] (wot) to (wot2wab) to (wab);
\draw[semithick] (wbb) to (wbb2ewb) to (ewb);
\draw[semithick] (awbb) to (awbb2wbb) to (wbb);
\draw[semithick] (awbb) to (awbb2wot) to (wot);
\draw[semithick] (wcb) to (wcb2wbb) to (wbb);

\draw[thick] (-6em,4em) -- (-6em,-4em);
\draw[thick] (2em,4em) -- (2em,-4em);
\end{tikzpicture}
\caption{A hierarchy on semantic equivalence relations for {\piCal} processes,
with separation lines indicating
where the encodings discussed in this paper pass and fail validity.}
\label{fig:semeq_hier}
\end{figure}
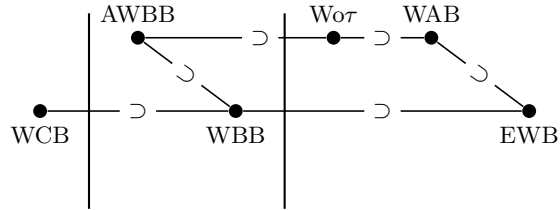

\section{Validity up to an Equivalence versus Validity \`a la Gorla}\label{sec:versus}

The idea of introducing a success predicate $\surd$ to the source and target language of an
encoding, as implicit in Gorla's criterion of success sensitiveness, can be applied to the
equivalence based approach as well.

\begin{definition}\rm\label{df:success bis}
Let $\mathcal{P}_{\rm s},~\mathcal{P}_{\rm t}$ be languages equipped with a reduction relation
$\longmapsto$, and \plat{$\mathcal{P}_{\rm s}^\surd,\,\mathcal{P}_{\rm t}^\surd$} their extensions with a success predicate $\surd$.
A symmetric relation {\mrel} on \plat{$\T\mathbin{:=}\T_{\rm s}^\surd \uplus \T_{\rm t}^\surd$} is a {\em success respecting weak reduction bisimulation}
if $P\mathbin{\mrel} Q$ implies
\begin{enumerate}
\item if $P \sbarb \surd$ then $Q \ocomm \surd$ and
\item if $P \longmapsto P'$ then a $Q'$ exists with
  $Q \Longmapsto Q'$ and $P'\mathop{\mrel} Q'$.%
\end{enumerate}
The largest success respecting weak reduction bisimulation is denoted \plat{$\wssbb$}.

An compositional encoding $\fT:\T_{\rm s} \rightarrow \T_{\rm t}$ is \emph{valid up to \plat{$\wssbb$}}
if its extension \plat{$\fT^\surd_{~}:\T_{\rm s}^\surd \rightarrow \T_{\rm t}^\surd$}, defined by \plat{$\fT^\surd(\surd)\mathbin{:=}\surd$},
satisfies \plat{$\fT^\surd(P) \wssbb P$} for all \plat{$P\mathbin\in \T_{\rm s}^\surd\!$}.
\end{definition}
Trivially, a variant of \lem{barbs} with $\surd$ in the role of $a$ holds for $\fT_{\rm B}$ as well
as $\fTHT$: we have $P \scomm \surd$ iff $\fT_{\rm B}(P) \scomm \surd$ iff $\fTHT(P) \scomm \surd$.
Using this, the material in the \hyperlink{appendix}{appendix} implies that:

\begin{theorem}\rm\label{thm:wssbb}
  The encodings $\fT_{\rm B}$ and $\fTHT$ are valid up to \plat{$\wssbb$}.
  \qed
\end{theorem}
This approach has the distinct advantage over dealing with input and output barbs that both
encodings are seen to be valid without worrying on what kinds of barbs to use exactly.

The following correspondence between operational correspondence, success sensitivity and validity up
to \plat{$\wssbb$} was observed in \cite{EPTCS190.4}, and not hard to infer from the definitions.

\begin{theorem}\rm\label{thm:versus}
  An encoding $\fT$ is success sensitive\vspace{-2pt} and satisfies operational correspondence criteria
  $(\mathfrak{C}')$ and $(\mathfrak{W})$, taking $\asymp$ to be $\wssbb$, iff it is valid up to $\wssbb$.
  \qed
\end{theorem}
This yields the result promised in \sect{operational correspondence}:
\begin{corollary}\rm
  The encodings $\fT_{\rm B}$ and $\fTHT$ satisfy criterion $(\mathfrak{W})$.
\end{corollary}
The validity of $\fT_{\rm B}$ and $\fTHT$ by Gorla's criteria, established in \cite{vG18a},
by the analysis of \cite{EPTCS190.4}, already implied that $\fT_{\rm B}$ and $\fTHT$ are valid up to
\emph{success respecting$\,$ coupled reduction similarity}~\cite{EPTCS190.4}, a semantic equivalence
strictly coarser than \plat{$\wssbb\!\!$}.\linebreak
\thm{wssbb} yields a nontrivial strengthening of that result.

Gorla's criterion of divergence reflection can be strengthened to \emph{divergence preservation}
by requiring\vspace{-2ex} $$\trans{S} \longmapsto_{\rm t}^\omega ~\Leftrightarrow~ S \longmapsto_{\rm s}^\omega\;;$$
by \cite[Remark 1]{vG18a} this criterion is satisfied by  $\fT_{\rm B}$ and $\fTHT$ as well.
A bisimulation $\mrel$ is said to preserve divergence iff $P\mathbin{\mrel} Q$ implies
$P \longmapsto_{\rm t}^\omega ~\Leftrightarrow~ Q \longmapsto_{\rm s}^\omega$;
the largest divergence preserving, success respecting weak reduction bisimulation is denoted \plat{$\wssbbd$}.
As observed in \cite{EPTCS190.4}, \thm{versus} can be extended as follows with divergence
preservation:
\begin{observation}\rm
  An encoding $\fT$ is success sensitive, divergence preserving,\vspace{-2pt} and satisfies operational correspondence criteria
  $(\mathfrak{C}')$ and $(\mathfrak{W})$, taking $\asymp$ to be $\wssbbd$,\vspace{-2pt} iff it is valid up to $\wssbbd$.
  \qed
\end{observation}
Hence, $\fT_{\rm B}$ and $\fTHT$ are valid up to \plat{$\wssbbd$}.
This statement implies all criteria of Gorla, except for name invariance.

In \cite[Definition~26]{vG18e} the notion of \emph{divergence preserving branching barbed bisimilarity} is defined.
This definition is parametrised by the choice of barbs; when taking the success predicate $\surd$ as
only barb, it could be called \emph{divergence preserving, success respecting branching reduction bisimilarity}.\vspace{-2pt}
It is strictly finer then $\wssbbd$. It is not hard to adapt the proof of \thm{boudolwrequ} in the
\hyperlink{appendix}{appendix} to show that $\fT_{\rm B}$ and $\fTHT$ are even valid up to this equivalence.

\section{Full Abstraction}\label{sec:full abstraction}

The criterion of full abstraction is parametrised by the choice of two semantic equivalences $\sim_{\rm S}$ and $\sim_{\rm T}$,
one on the source and one on the target language. It requires, for source expressions $P$ and $Q$, that
$P \sim_{\rm S} Q \Leftrightarrow \fT(P) \sim_{\rm T} \fT(Q)$.

It is well known that the encodings $\fT_{\rm B}$ and $\fTHT$ fail to be fully abstract w.r.t.\ $\cong^c$ and $\cong^c_a$.
Here $\cong^c$ is \emph{weak barbed congruence}, the congruence closure of \plat{$\wssbb$} (or $\wbbisim$) on the source language, and
$\cong^c_a$ is \emph{asynchronous weak barbed congruence}, the congruence closure of \plat{$\wssbb$}
(or $\wbbisim$) on the target language.
These are often deemed to be the most natural semantic equivalences on $\pims$ and $\pima$.
The well-known counterexample is given by the $\pims$ processes $\bar x z |  \bar x z$ and $\bar xz.\bar x z$.
Although related by $\cong^c$, their translations are not related by $\cong^c_a$.

In \cite{DYZ18} this problem is addressed by proposing a strict subcalculus $\it SA\pi$ of the 
target language that contains the image of the source language under of a version Honda
\& Tokoro's encoding, such that this encoding is fully abstract w.r.t.\ $\cong^c$ and
the congruence closure of \plat{$\wssbb$} (or $\wbbisim$) w.r.t.\ $\it SA\pi$.
In \cite{QW00} a similar solution to the same problem was found earlier, but for a variant of
Boudol's encoding from the \emph{polyadic} $\pi$-calculus to the (monadic) asynchronous
$\pi$-calculus.  They define a class of \emph{well-typed} expressions in the asynchronous
$\pi$-calculus, such that the well-typed expressions constitute a subcalculus of the target language
that contains the image of the source language under the encoding. Again, the encoding is fully
abstract w.r.t.\ $\cong^c$ and the congruence closure of \plat{$\wssbb\!\!$} (or $\wbbisim$) w.r.t.\ that sublanguage.

By \cite[Theorem 4]{vG18b} such results can always be achieved, namely by taking
as target language exactly the image of the source language under the encoding.
In this sense a full abstraction result is a direct consequence of the validity of the encodings up
to \plat{$\wssbb$}, taking for $\sim_{\rm S}$ the congruence closure of \plat{$\wssbb$} w.r.t.\ the source language,
and for $\sim_{\rm T}$ the congruence closure of \plat{$\wssbb$} w.r.t.\ the image of the source language
within the target language.
What the results of \cite{QW00,DYZ18} add is that the sublanguage may be strictly larger than the
image of the source language, and that its definition is not phrased in terms of the encoding.

\section{Conclusion}\label{sec:conclusion}

We examined which of the quality criteria for encodings from the literature support the validity
of the well-known encodings $\fT_{\rm B}$ and $\fTHT$ of the asynchronous into the synchronous $\pi$-calculus.
It was already known \cite{vG18a} that these encodings are valid \`a la Gorla \cite{Gorla10a}; this implies that
they are valid up to success respecting coupled reduction similarity \cite{EPTCS190.4}.
We strengthened this result by showing that they are even valid up to divergence preserving,
success respecting weak reduction bisimilarity. That statement implies all criteria of Gorla, except for name invariance.
Moreover, it implies a stronger form of operation soundness then considered by Gorla, namely
\begin{center}
\mbox{}\hfill
if $\trans{S} \Longmapsto_{\rm t} \pT$ then $\exists S'\!:$ $S\Longmapsto_{\rm s} S'$ and $\pT\asymp_{\rm t} \trans{S'}$.\hfill~~$(\mathfrak{W})$
\end{center}
Crucial for all these results is that we employ Gorla's external barb $\surd$, a success predicate
on processes. When reverting to the internal barns $x$ and $\bar x$ commonly used in the
$\pi$-calculus, we see a potential difference in quality between the encodings $\fT_{\rm B}$ and $\fTHT$.
Boudol's translation $\fT_{\rm B}$ is valid up to weak barbed bisimilarity, regardless whether all
barbs are used, or only output barbs $\bar x$. However, Honda and Tokoro's translation $\fTHT$
is not valid under either of these forms of weak barbed bisimilarity.
In order to prove the validity of $\fTHT$, we
had to use the novel weak channel bisimilarity that does not distinguish between input and output channels.
Conversely, we conjecture that there is no natural equivalence for which $\fTHT$ is valid, but $\fT_{\rm B}$ is not.
Hence, Honda and Tokoro's encoding can be regarded as weaker than the one of Boudol.
Whether $\fT_{\rm B}$ is to be preferred, because it meets stronger
requirements/equivalences, is a decision that should be driven by
the requirements of an application the encoding is used for.

The validity of $\fT_{\rm B}$ under semantic equivalences has earlier been investigated in \cite{CC01,CCP07},
In \cite{CC01} it is established that $\fT_{\rm B}$ is valid up to \emph{may testing}
\cite{DH84} and \emph{fair testing equivalence} \cite{BRV95,NC95}. Both results now follow from
\thm{wssbb}, since may and fair testing equivalence are coarser then \plat{$\wssbb$}.
On the other hand, \cite{CC01} also shows that $\fT_{\rm B}$ is not valid up to a
form of \emph{must testing}; in \cite{CCP07} this result is strengthened to pertain to any encoding
of $\pims$ into $\pima$.\vspace{1pt} It follows that this form of must testing equivalence is not implied by
\plat{$\wssbb$}, and not even by \plat{$\wssbbd$}.

\bibliographystyle{eptcs}
\renewcommand{\doi}[1]{doi:\urlalt{http://dx.doi.org/#1}{#1}}
\bibliography{../../../../Stanford/lib/abbreviations,../../../../Stanford/lib/dbase,ipl}
\vspace{2ex}

\appendix

\renewcommand{\sI}{\fT_{\rm B}}

\noindent
\hypertarget{appendix}{{\large\bf Appendix: Boudol's Translation is Valid up to $\wbb$}}\vspace{3.5ex}

\noindent
Before we prove validity of Boudol's translation up to weak barbed bisimulation, we further investigate the protocol steps established by Boudol's encoding.
Let $P'\mathbin=\bar{x}z.P$ and $Q'\mathbin=x(y).Q$. Pick $u,v$ not free in $P$ and $Q$, with
$u\mathbin{\neq} v$. Write $P^*\mathbin{:=}\bar{v}z | \sI(P)$ and $Q^*\mathbin{:=}v(y).\sI(Q)$. Then
\[\begin{array}{rcl}\sI(P'|Q') &=& (u)(\bar{x}u | u(v).P^*) ~|~ x(u).(v)(\bar{u}v | Q^*)\\
&\longmapsto& (u)\big(u(v).P^* ~|~ (v)(\bar{u}v | Q^*)\big)\\
&\longmapsto& (v)(P^* | Q^*)\\
&\longmapsto& \sI(P) | (\sI(Q)\subs{z}{y})\;.
\end{array}\]
Here structural congruence is applied in omitting parallel components $\bm{0}$ and empty binders
$(u)$ and $(v)$. Now the crucial idea in our proof is that the last two reductions are \emph{inert},
in that set of the potential behaviours of a process is not diminished by doing (internal) steps of
this kind. The first reduction above in general is not inert, as it creates a commitment between a
sender and a receiver to communicate, and this commitment goes at the expense of the potential of
one of the two parties to do this communication with another partner.
We employ a relation that captures these inert reductions in a context.

\newcommand{\XX}{P}
\newcommand{\YY}{Q}
\begin{definition}[\cite{vG18a}]\rm\label{def:cool-rel}
Let $\equiv\!\Rrightarrow$ be the smallest relation on $\bbT_{\pima}$ such that
\begin{enumerate}
\item $(v)(\bar vy | \XX | v(z).\YY) \equiv\!\Rrightarrow \XX|(\YY\subs{y}{z})$,
\item if $\XX \equiv\!\Rrightarrow \YY$ then $\XX|C \equiv\!\Rrightarrow \YY|C$,
\item if $\XX \equiv\!\Rrightarrow \YY$ then $(w) \XX \equiv\!\Rrightarrow (w) \YY$,
\item if $\XX \equred \XX' \equiv\!\Rrightarrow \YY' \equred \YY$ then $\XX \equiv\!\Rrightarrow \YY$,
\end{enumerate}
where
$v \not\in \fn(\XX) \cup \fn(\YY\subs{y}{z})$.
\end{definition}

\noindent
First of all observe that whenever two processes are related by $\equiv\!\Rrightarrow$, an actual
reduction takes place.

\begin{lemma}[\cite{vG18a}]\label{lem1}\rm
If $P \equiv\!\Rrightarrow Q$ then $P \longmapsto Q$.
\end{lemma}

\noindent
The next two lemmas confirm that inert reductions do not diminish the
potential behaviour of a process.
\begin{lemma}[\cite{vG18a}]\label{lem2}\rm
If $P \equiv\!\Rrightarrow Q$ and $P \longmapsto P'$ with $P'\not\equred Q$
then there is a $Q'$ with $Q \longmapsto Q'$ and $P' \equiv\!\Rrightarrow Q'$.
\end{lemma}
\begin{corollary}\label{lem2 transitive}\rm
If $P \equiv\!\Rrightarrow^* Q$ and $P \longmapsto P'$
then either $P'\mathbin{\equiv\!\Rrightarrow^*} Q$ or there is a $Q'$ with $Q \longmapsto Q'$ and $P' \equiv\!\Rrightarrow^* Q'$.
\end{corollary}
\begin{proof}
By repeated application of Lemma~\ref{lem2}.
\qed
\end{proof}

\begin{lemma}\label{postponed barbs}\rm
If $P \equiv\!\Rrightarrow Q$ and $P \scomm a$ for $a \in \{ x, \bar x \,|\, x \in \mathcal N \}$
then $Q \scomm a$.
\end{lemma}
\begin{proof}
\newcommand{\rn}{z}
\newcommand{\un}{y}
Let $(\tilde w)P$ for $\tilde w \mathbin=\{w_1,\dots,w_n\}\mathbin\subseteq \mathcal{N}$
with $n\mathbin\in\mbox{\bbb N}$ denote $(w_1)\cdots(w_n)P$ for some
arbitrary order of the $(w_i)$.
Using a trivial variant of Lemma 1.2.20 in \cite{SW01book},
there are
$\tilde w \subseteq\mathcal{N}$, $x,\un,\rn\mathbin\in\mathcal{N}$
and $R,C\mathbin\in\bbT_{\pima}$, such that $x\in\tilde w$ and
$P \mathbin{\equred} (\tilde w)((\bar x \un | x(\rn) . R) | C ) \mathbin{\longmapsto}
(\tilde w)(( \bm{0} | R\subs{\un}{\rn}) | C ) \mathbin{\equred} Q$.
Since $P{\sbarb{a}}$, it must be that $a{=}u$ or $\bar u$ with $u\notin\tilde w$, and
$C{\sbarb{a}}$. Hence $Q{\sbarb{a}}$.
\qed
\end{proof}

\noindent
The following lemma states, in terms of Gorla's framework, {\em operational completeness}~\cite{Gorla10a}:
if a source term is able to make a step, then its translation is able to simulate that step by protocol steps.

\begin{lemma}[\cite{vG18a}]\label{lem5}
Let $P,P' \mathbin\in \bbT_{\pims}$.
If $P \longmapsto P'$ then $\sI(P) \longmapsto^* \sI(P')$.
\end{lemma}

\noindent
Finally, the next lemma was a crucial step in establishing \emph{operational soundness}~\cite{Gorla10a}.

\begin{lemma}[\cite{vG18a}]\label{lem6}\rm
Let $P \mathbin\in \bbT_{\pims}\!$ and $Q \mathbin\in \bbT_{\pima}\!$.
If $\sI(P) \mathbin{\longmapsto} Q$ then there is a $P'$ with $P \longmapsto P'$
and $Q\equiv\!\Rrightarrow^* \sI(P')$.
\end{lemma}

\noindent
Using these lemmas, we prove the validity of Boudol's encoding up to weak barbed bisimilarity.

\begin{theorem}\label{thm:boudolwrequ}\rm
Boudol's encoding is valid up to $\wbb$.
\end{theorem}

\newcommand{\rel}{\mrel}
\begin{proof}
Define the relation $\rel$ by $P\rel Q$ iff $Q \mathrel{\equiv\!\Rrightarrow^*} \sI(P)$.
It suffices to show that the symmetric closure of $\rel$ is a weak barbed bisimulation.

To show that $\rel$ satisfies Clause 1 of \df{WBB},
suppose $P \rel Q$ and $P \scomm a$ for $a \in \{ x, \bar x \,|\, x \in \mathcal N \}$.
Then $\sI(P) \scomm a$ by \lem{barbs}.
Since  $Q \mathrel{\equiv\!\Rrightarrow^*} \sI(P)$, we obtain $Q \longmapsto^* \sI(P)$ by Lemma~\ref{lem1}, and
thus $Q \wcomm a$.\pagebreak[3]

To show that $\rel$ also satisfies Clause 2,
suppose $P \rel Q$ and $P \longmapsto P'$. Since $Q \mathrel{\equiv\!\Rrightarrow^*} \sI(P)$, by
Lemmas~\ref{lem1} and \ref{lem5} we have $Q \mathbin{\longmapsto^*} \sI(P) \mathbin{\longmapsto^*} \sI(P')$, and also $P' \rel \sI(P')$.

To show that $\rel^{-1}$ satisfies Clause 1, suppose $P \rel Q$ and $Q \scomm a$.
Since $Q \mathrel{\equiv\!\Rrightarrow^*} \sI(P)$, Lemma~\ref{postponed barbs} yields \mbox{$\sI(P) \scomm a$},
and \lem{barbs} gives $P \scomm a$, which implies $P \wcomm a$.

To show that $\rel^{-1}$ satisfies Clause 2, suppose
$P \rel Q$ and $Q \longmapsto Q'$. Since $Q \mathbin{\equiv\!\Rrightarrow^*} \sI(P)$, by
Corollary~\ref{lem2 transitive} either $Q' \equiv\!\Rrightarrow^* \sI(P)$
or there is a $Q''$ with $\sI(P) \longmapsto Q''$ and $Q' \mathbin{\equiv\!\Rrightarrow^*} Q''\!$.
In the first case $P \rel Q'$, so taking $P':=P$ we are done.
In the second case, by Lemma~\ref{lem6} there is a $P'$ with $P \longmapsto P'$ and $Q'' \equiv\!\Rrightarrow^* \sI(P')$.
We thus have $P' \rel Q'$.
\qed
\end{proof}
\end{document}